\documentclass[twoside,11pt,reqno]{amsart}
\usepackage{amsmath}
\usepackage{amsfonts}
\usepackage{amssymb}
\usepackage{mathrsfs} 
\usepackage[toc]{appendix}
\usepackage{pstricks}
\usepackage{pst-node,pst-3d,pst-text,pst-tree}
\usepackage{graphicx}
\usepackage{color}
\usepackage{array}
\usepackage{microtype}
\usepackage{caption}
\usepackage{enumerate}

\theoremstyle{plain}
\newtheorem{thm}{Theorem}

\newtheorem{cor}{Corollary}

\newtheorem{lem}{Lemma}

\theoremstyle{definition}
\newtheorem{mydef}{Definition}

\newcommand{\beqn}{\begin{eqnarray}\begin{aligned}}
\newcommand{\eqn}{\end{aligned}\end{eqnarray}}

\newcounter{mycnt}

\def\ov{\overline}
\def\1{\ov{1}}
\def\2{\ov{2}}
\def\3{\ov{3}}

\definecolor{mygreen}{rgb}{0.2656,0.5039,0.2148}
\definecolor{myred}{rgb}{0.75,0,0.25}

\begin{document}

\title[Phylogenetic dimensional reduction ]
{Dimensional reduction for the general Markov model on phylogenetic trees} 


%
\author{Jeremy G Sumner}
\address{J G Sumner,  %
School of Physical Sciences, Mathematics, University of Tasmania, 
Private Bag 37, GPO, Hobart Tas 7001, Australia}
\email{Jeremy.Sumner@utas.edu.au}
%


\keywords{}

\date{\today}

\begin{abstract}
We present a method of dimensional reduction for the general Markov model of sequence evolution on a phylogenetic tree.
We show that taking certain linear combinations of the associated random variables (site pattern counts) reduces the dimensionality of the model from exponential in the number of extant taxa, to quadratic in the number of taxa, while retaining the ability to statistically identify phylogenetic divergence events.
A key feature is the identification of an invariant subspace which depends only bilinearly on the model parameters, in contrast to the usual multi-linear dependence in the full space.
We discuss potential applications including the computation of split (edge) weights on phylogenetic trees from observed sequence data. 
\end{abstract}
\maketitle

\section{Introduction}

Phylogenetics is the suite of mathematical and computational methods available to biologists for the inference of the evolutionary history of extant species.
Typical data inputs to these methods are molecular sequences in various forms, including DNA, sequences of amino acids, codons, and proteins (usually first aligned and then summarised into site pattern counts), and even gene orderings at the scale of whole genomes. 
Underlying many modern methods (particularly likelihood and Bayesian approaches) is a stochastic model of molecular state evolution; most typically constructed as a Markov process on candidate evolutionary trees.
The reader is referred to the texts \cite{felsenstein2004inferring,semple2003phylogenetics,yang2014molecular} for excellent backgrounds on the mathematical, statistical, and computational aspects of phylogenetics.

Over the last two decades or so there has been much mathematical progress emphasizing natural \emph{algebraic} structures that underlie phylogenetic models.
These range from algebraic geometry arising from model parameterization maps \cite{allman2008phylogenetic,sturmfels2005toric}, to the combinatorial group theory inherent in genome rearrangement models \cite{francis2014algebraic}, to methods of discrete Fourier transforms \cite{hendy1994discrete}.  
Concurrently, our research has explored the role of (matrix) Lie groups \cite{sumner2012lie} together with their associated representation theory and invariant functions \cite{jarvis2015matrix,sumner2008markov,sumner2005entanglement}, in particular defining what we refer to as ``Markov invariants''\footnote{Not to be confused with ``phylogenetic invariants''; the distinction will be discussed below.} (we provide a recent review in \cite{jarvis2014adventures}). 

In their most abstract setting, Markov invariants provide the means to reduce the high dimensionality of phylogenetic models (determined by the number of observable site patterns) from exponential in the number of taxa to one-dimensional subspaces spanned by individual polynomials (from the statistical point of view, these are functions on the random variables given by site pattern counts in molecular sequence alignments).
Unfortunately, due to algebraic combinatorial explosion, it is only feasible to explore these polynomials explicitly for small cases, and this has only been achieved so far for four taxa --- the so-called ``quartet'' case \cite{holland2013low,sumner2009markov}.

Even though the mathematical properties of Markov invariants are rather elegant, and in themselves quartets are of fundamental importance to phylogenetics, the practical utility of this approach becomes questionable when biologists are interested in the inference of large phylogenetic trees with many, possibly thousands, of taxa.
For these and other reasons, it is perhaps not surprising that the emergence of algebraic approaches has not been accompanied by an upsurge in usage by practicing biologists.

In this paper, we will explore a generalization which compromises on the dimensional reduction provided by Markov invariants slightly: we reduce the dimensionality of the model from exponential in the number of taxa to \emph{quadratic} in the number of taxa.
This reduction comes with the significant benefit that the method we describe is generally applicable to any number of taxa using standard computational techniques currently available in most programming languages (linear transformations and singular value decomposition SVD), with no further theoretical work required.
The approach hinges on a simple algebraic observation regarding the embedding of Markov matrices into the matrix group of linear invertible affine transformations, and the identification of a particular invariant subspace after distributing over Kronecker (or tensor) products.

The paper proceeds as follows.
In \S\ref{sec:background} we give the required background on the algebraic structures underlying phylogenetic models.
Specifically, we discuss the so-called ``flattenings'' and associated rank properties.
In \S\ref{sec:dimreduct} we present our main result of dimensional reduction for phylogenetic models.
In \S\ref{sec:discuss} we discuss potential strategies for applying the dimensional reduction in the practical setting of phylogenetic tree inference.

\section{Background}
\label{sec:background}

We assume we have $L$ biological taxa labelled as $1,2,\ldots,L$ together with homologous (aligned) molecular states drawn from a finite state space $\kappa$.
For example, we have $|\kappa|\!=\!4$ for DNA and $|\kappa|\!=\!20$ for amino acids.
Since our results are independent of the size of the state space, we will generically write $|\kappa|\!=\!k$.
To produce a phylogenetic model, we fix (i) a (rooted) binary tree $\mathcal{T}$ with $L$ leaves (degree 1 vertices), (ii) an initial probability distribution $(\pi_i)_{i\in \kappa}$, and (iii) Markov matrices  $M_{e}$ for each edge  $e\in \mathcal{T}$. (Throughout this paper, we will follow the convention that Markov matrices have unit-\emph{column} sums, as opposed to unit-row sums.)
These inputs then produce a probability distribution of site patterns $P=(p_{i_1i_2\ldots i_L})_{i_j\in\kappa}$ at the leaves of the tree and we refer to any such distribution as having ``arisen'' under the general Markov model on the tree $\mathcal{T}$.
This process is illustrated in Figure~\ref{fig:7leaf}.

\begin{figure}[htb]
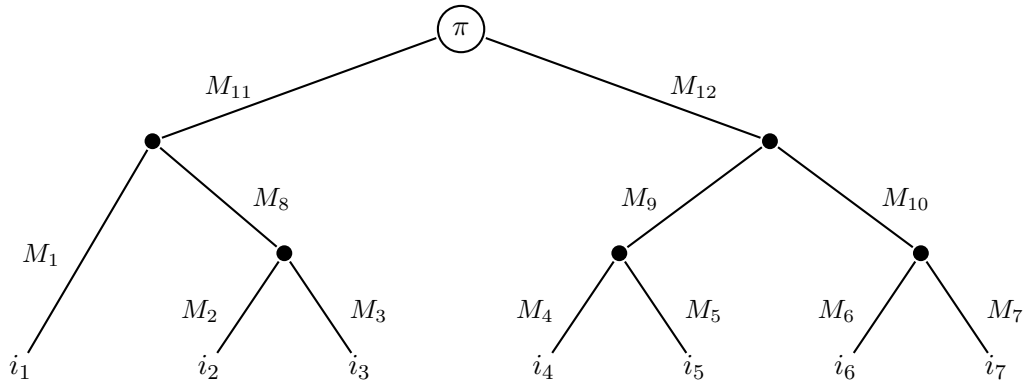

\caption{Graphical illustration of the construction of a site pattern distribution $P=(p_{i_1i_2i_3i_4i_5i_6i_7})_{i_j\in\kappa}$ on a phylogenetic tree $\mathcal{T}$ with seven taxa. Input data is a set of Markov matrices $\{M_{e}\}_{e\in \mathcal{T}}$ and an initial distribution $(\pi_i)_{i\in \kappa}$.}
\vspace{1em}
    \centering
    \pstree[levelsep=9ex,treesep=1.6cm,nodesep=.8pt]{\Tcircle{$\pi$}}{%
        \pstree{\TC*[radius=3pt] \tlput{\small$M_{11}$}}{
            \skiplevel{\TR{$i_1$} \tlput{\small$M_{1}$}}
            \Tn
            \pstree{\TC*[radius=3pt]\trput{\small$M_{8}$}}{
                \TR{$i_2$} \tlput{\small$M_{2}$}
                \TR{$i_3$} \trput{\small$M_{3}$}
            }
        }
        \Tn \Tn \Tn \Tn
        \pstree{\TC*[radius=3pt] \trput{\small$M_{12}$}}{
            \pstree{\TC*[radius=3pt] \tlput{\small$M_{9}$}}{
                \TR{$i_4$} \tlput{\small$M_{4}$}
                \TR{$i_5$} \trput{\small$M_{5}$}
            }
            \Tn
            \pstree{\TC*[radius=3pt] \trput{\small$M_{10}$}}{
                \TR{$i_6$} \tlput{\small$M_{6}$}
                \TR{$i_7$} \trput{\small$M_{7}$}
            }
        }
    }
    \label{fig:7leaf}
\end{figure}

Even though such a model specifies a location for the root of the tree, it is generally the case that the choice of root either does not affect the output distribution in any way \cite{felsenstein1981evolutionary}, or for each alternative root position there is a corresponding alternative choice of model parameters which produce the same distribution. 
Thus, without loss of generality, we may remove the root vertex from $\mathcal{T}$ and hitherto consider $\mathcal{T}$ to be a binary tree with all vertices of degree 3 or 1.
For simplicity we will refer to both the vertices of degree 1 \emph{and} their adjacent edges as ``leaves'' (context will resolve any ambiguity).


Consider an $m|n$ ``split'' $A|B$ of the taxon set $\{1,2,\ldots,L\}$.
This is defined as a bipartition $\{A,B\}$ with $A,B\subseteq\{1,2,\ldots,L\}$, $A\cap B=\emptyset$, $A\cup B=\{1,2,\ldots,L\}$, $|A|\!=\!m$, and $|B|\!=\!n$, so $m+n\!=\!L$.
Suppose we fix an (unrooted) binary tree $\mathcal{T}$ with leaf set $\{1,2,\ldots,L\}$.
Then each edge in $\mathcal{T}$ gives rise to a split $A|B$ via deletion of the edge and the resulting partitioning of the leaf vertices.
Correspondingly, for each split $A|B$ we may define a so-called ``flattening'' of the array $P=(p_{i_1i_2\ldots i_L})_{i_j\in\kappa}$, as follows\footnote{Illustrative examples of flattenings are given in the introduction of \cite{allman2008phylogenetic}.}.
Supposing $A\!=\!\{s_1,s_2,\ldots,s_m\}$ and $B\!=\!\{t_1,t_2,\ldots,t_n\}$, the flattening $\text{Flat}_{A|B}(P)$ is the $k^m\times k^n$ matrix with rows indexed by the strings $i_{s_1}i_{s_2}\ldots i_{s_m}$, columns indexed by strings $i_{t_1}i_{t_2}\ldots i_{t_n}$, and entries $p_{i_1i_2\ldots i_L}$.
Although the row and column indices may be ordered in any way we choose without affecting the present discussion, it is convenient to take the ordering which is consistent with the tensor/Kronecker products discussed in the next section.

This construction gives rise to the so-called ``edge invariants'' since it can be shown that
\emph{the flattening $\text{Flat}_{A|B}(P)$ has matrix rank at most $k$ if $A|B$ is a split in the tree and at least $k^2$ otherwise}. 
(These facts are established in \cite{allman2008phylogenetic,eriksson2005tree} for ``generic'' cases --- weak conditions on model parameters need to be assumed.)
More generally, \cite{eriksson2005tree} gives an algorithm to find $r\in\mathbb{N}$ such that the rank of the flattening is at most $k^r$ and, up to this bound, maximal rank is attained in generic cases.

The edge invariants are defined as the $(k+1)\!\times\!(k+1)$ minors of the flattening $\text{Flat}_{A|B}(P)$, which are seen to vanish if the edge corresponding to $A|B$ occurs in $\mathcal{T}$.
These minors form polynomials in the variables $(p_{i_1i_2\ldots i_L})_{i_j\in \kappa}$ and are specific cases of what are known as ``phylogenetic invariants'' in the literature \cite{cavender1987invariants,lake1987rate}, which are defined as polynomials that vanish on probability distributions arising from certain subsets of phylogenetic trees (see \cite{allman2008phylogenetic,casanellas2011relevant,draisma2009ideals} for more recent developments).
For our purposes, it is important to distinguish ``phylogenetic'' from ``Markov'' invariants; these are formally distinct concepts with instances that sometimes, but not always, coincide.

The flattenings and/or edge invariants may be used as a simple test for the occurrence of the edge corresponding to $A|B$ in the tree $\mathcal{T}$ as follows.
\begin{enumerate}
\item Collect the observed molecular sequences into an alignment (using standard, freely available software) to produce a site pattern count array $F=(f_{i_1i_2\ldots i_L})_{i_j\in\kappa}$, where $f_{i_1i_2\ldots i_L}$ is the number of sites which have state pattern $i_1i_2\ldots i_L$;
\item  Suppose that $F$ is a sample $F\sim \text{multinomial}(P,N)$ where $P$ is some pattern distribution arising under the general Markov model of sequence evolution on an (unknown) tree $\mathcal{T}$ and $N$ is the length of the alignment;
\item Test for suitability of the edge corresponding to $A|B$ belonging to $\mathcal{T}$ by computing the rank of $\text{Flat}_{A|B}(F)$; 
\item If this rank is ``close'' to $k$, conclude that $A|B$ is in the tree $\mathcal{T}$.
\end{enumerate}

Ignoring issues of how to numerically estimate the rank (\cite{eriksson2005tree} suggests using the SVD decomposition, and \cite{casanellas2007performance} have explored using the minors directly), or how one decides what we mean by ``close'' in a statistically justified way (this is an open problem generally), as a first step this process provides a mathematically elegant test for the evidence of a specific edge in a phylogenetic tree.
Significantly, this test is statistically valid no matter what Markov model of molecular state evolution is presupposed (since it is valid for the ``general'' Markov model).
On the down side, the complexity of the rank estimation and comparison problem is exponential in the number of taxa $L$ since $\text{Flat}_{A|B}(F)$ is a $k^m\times k^n$ matrix.
For a large number of taxa $L\!=\!m+n$ and finite sequence length, the flattening matrix will clearly be sparse.

The SVD methodology has implemented by \cite{fernandez2015invariant} and \cite{chifman2014quartet} with impressive results.
However, these studies restrict attention to quartets only.
Most recently \cite{allman2016split} presented an efficient implementation of the flattenings and SVD for arbitrary tree sizes using sparse matrix computation.

\section{Dimensional reduction}
\label{sec:dimreduct}

As has been observed in many places \cite{allman2008phylogenetic,bashford2004u,bryant2009hadamard,sumner2012algebra}, phylogenetic models implemented as a Markov chain on a binary tree can be described algebraically using tensor product spaces and/or Kronecker product operations.
Following the discussion given in \cite{sumner2008markov}, if we fix a site pattern distribution $P$ obtained from the general Markov model on a phylogenetic tree as described in the previous section,  we may define an alternative distribution $\tilde{P}$ specified by setting all Markov matrices on the leaves of the tree equal to the identity.
As is shown in \cite{sumner2008markov}, using tensor notation we may write
\beqn
\label{eq:leafextend}
P=M_1\otimes M_2\otimes \ldots \otimes M_L\cdot \tilde{P}.
\eqn

Written in terms of an $m|n$ flattening on the split $A|B$, this equation is expressible using matrix products:
\[
\text{Flat}_{A|B}(P)=M_A\text{Flat}_{A|B}(\tilde{P})M_B^T,
\]
where  $M_A$ is the $k^m\times k^m$ matrix obtained by taking the Kronecker product of the Markov matrices on the leaves belonging to $A$, and  $M_B$ is the $k^n\times k^n$ matrix obtained by taking the Kronecker product of the Markov matrices on the leaves belonging to $B$.

Further, given that Markov matrices are multiplicative (that is, multiplication of two Markov matrices gives another Markov matrix), we may generalise (\ref{eq:leafextend}) to be understood as the transformation rule 
\beqn
\label{eq:leafactionP}
P\rightarrow M_1\otimes M_2\otimes \ldots M_L\cdot P,
\eqn
expressing how a distribution $P$ arising from the general Markov model on a phylogenetic tree changes if additional Markov matrices $M_i$ are applied at the leaves of the tree.
Expressed in terms of the flattening, this transformation rule becomes
\beqn
\label{eq:leafactionG}
\text{Flat}_{A|B}(P)\rightarrow M_A\text{Flat}_{A|B}(P)M_B^T.
\eqn 


Practical importance of this transformation rule follows from the observation that the uniqueness of the binary tree underlying the phylogenetic model is unaffected by the choice of Markov matrices on the leaves (even under the condition that they are set equal to the identity matrix), and this is highlighted by the fact that the corresponding statement is not true for the internal edges of the tree.
To see this, one may refer to Figure~\ref{fig:7leaf} where restricting the matrix $M_8$ to be equal to the identity would imply, as judged by the distribution $P$, the split $\{2,3\}|\{1,4,5,6,7\}$ is indistinguishable from the alternative splits $\{1,2\}|\{3,4,5,6,7\}$, $\{1,3\}|\{2,4,5,6,7\}$, or $\{1,2,3\}|\{4,5,6,7\}$.
Additionally, it is easily seen that, if the relevant Markov matrices are of full rank, \emph{the rank of the flattening is unaffected by the transformation rule (\ref{eq:leafactionG})}.
The general philosophy we follow is that the transformation rules (\ref{eq:leafactionP}) and (\ref{eq:leafactionG}) are foundational to algebraic approaches to identifying splits on phylogenetic trees.
This aligns with our previous work on Markov invariants \cite{sumner2008markov} which seeks to find polynomials on the variables $(p_{i_1i_2\ldots i_L})_{i_j\in \kappa}$ that are invariant (up to scaling) under the transformation rule (\ref{eq:leafactionP}).

In this section we show that we can significantly reduce the dimensionality of the matrices involved in the transformation rule (\ref{eq:leafactionG}) by implementing  a linear change of coordinates on the flattening $\text{Flat}_{A|B}(P)$ and identifying the ``sub-flattening'' $\widehat{\text{Flat}}'_{A|B}(P)$ as a certain sub-matrix thereof.
Crucially, the dimension of the sub-flattening $\widehat{\text{Flat}}'_{A|B}(P)$ is \emph{quadratic} in the number of taxa $L$ and the sub-flattening retains contrasting  rank conditions when the true tree $\mathcal{T}$ does or does not contain the edge corresponding to $A|B$.

Before we do this however, the reader should note that the sub-flattening  $\widehat{\text{Flat}}'_{A|B}(P)$ is distinct from the related concept of ``thin flattenings'' introduced in \cite{casanellas2011relevant}, which are specific to each Markov model in the ``equivariant'' hierarchy presented in \cite{draisma2009ideals}.
Although the components of a thin flattening have reduced dimension (relative to the dimension of the full flattening $\text{Flat}_{A|B}(P)$), their dimension remains exponential in the number of taxa.
Importantly, in the case of the general Markov model, the thin flattening is equal to the original flattening and hence does not produce a dimensional reduction.

To derive the sub-flattening $\widehat{\text{Flat}}'_{A|B}(P)$, we make the simple observation that, under an appropriate similarity transformation, the condition that Markov matrices have unit column-sums\footnote{If one prefers to use unit row-sum Markov matrices, an analogous construction is obtained by taking the transpose in what follows.} is converted into the condition that a certain row has $k\!-\!1$ zeros and a single 1. 
Concretely, we may choose the similarity transformation $S$ so any $k\times k$ Markov matrix takes the form
\beqn
\label{eq:simtrans}
M':=SMS^{-1}=
\left(
\begin{matrix}
T & v \\
0 & 1
\end{matrix}
\right),
\eqn
where $T$ is a $(k\!-\!1)\!\times\! (k\!-\!1)$ matrix and $v$ is a $k\!-\!1$ column vector.
Generically (and hence without loss of relevance to the applied setting), we may assume $\det(M)\neq 0$ which in turn implies $\det(T)\neq 0$.
The set of such matrices (\ref{eq:simtrans}) then forms what is known as the \emph{affine group} $\text{Aff}(k\!-\!1)$ \cite[Chap. 1]{baker2012matrix}.
This is the (matrix) Lie group of invertible affine linear transformations on $\mathbb{R}^{k-1}$ given by the semi-direct product $\text{GL}(k\!-\!1)\ltimes \mathbb{R}^{k\!-\!1}$. 
(This correspondence under similarity transformations between Markov matrices and the affine group seems to have first been made in \cite{johnson1985markov}.)

In applications it is always the case that the entries of Markov matrices $M$ are non-negative --- this after all ensures probability conservation of the Markov process.
However, we need not worry about these additional conditions to obtain the results we wish to derive here --- the above embedding of the transformed matrix $M'$ into the affine group is all we require.

Also noteworthy is the fact that there are infinitely many choices of similarity transformations $S$ which will achieve the form given above: any invertible matrix with a constant row will work.
This is not an issue for the theoretical results given in this paper, as they are independent of the particular choice of suitable similarity transformation $S$.
However, in applications where practical considerations become important (such as finite sequence lengths) certain choices of $S$ may outperform others.
Some further comments on these matters are given in the discussion below.

Under an $m|n$ flattening we have $M_A=M_{s_1}\otimes M_{s_2}\otimes \ldots \otimes M_{s_m}$ as a Kronecker product of the $m$ Markov matrices on the leaves $A$.
After taking the similarity transformation $M_{s_j}\mapsto M'_{s_j}=SM_{s_j}S^{-1}$ on each factor, an elementary computation shows that (possibly after some simultaneous row and column permutations) we can extract an $\left(m(k\!-\!1)\!+\!1\right)\times \left(m(k\!-\!1)\!+\!1\right)$ sub-block which takes the form
\beqn
\label{eq:subaction}
\widehat{M}'_A:=
\left(
\begin{matrix}
T_{i_1} & 0 & 0 & 0 & \cdots & v_{i_1} \\
0 & T_{i_2} & 0 & 0 & \cdots & v_{i_2}  \\
0 & 0 & T_{i_3} & 0 & \cdots & v_{i_3} \\
\vdots & & & \ddots  & & \vdots \\
0 & 0 & 0 & \cdots & T_{i_m} & v_{i_m} \\
0 & 0 & 0 & \cdots & 0 & 1\\
\end{matrix}
\right).
\eqn
Similarly, a corresponding expression exists for $\widehat{M}'_B$ as an $\left(n(k\!-\!1)\!+\!1\right)\times \left(n(k\!-\!1)\!+\!1\right)$ matrix.

As a simple example to convince the reader, consider the case of a two-fold Kronecker product:
\[
M'_{12}
\equiv
M'_1\otimes M'_2
=
\left(
\begin{matrix}
T_1 & v_1 \\
0 & 1
\end{matrix}
\right)
\otimes
\left(
\begin{matrix}
T_2 & v_2 \\
0 & 1
\end{matrix}
\right)
=
\left(
\begin{matrix}
T_1\otimes T_2 & T_1\otimes v_2  & v_1\otimes T_2 & v_1\otimes v_2 \\
0 & T_1  & 0 & v_1 \\
0 & 0 & T_2 & v_2 \\
0 & 0 & 0 & 1
\end{matrix}
\right),
\]
and define 
\[
\widehat{M}'_{12}:=
\left(
\begin{matrix}
T_1  & 0 & v_1 \\
0 & T_2 & v_2 \\
0 & 0 & 1
\end{matrix}
\right).
\]
Applying additional Kronecker products requires some simultaneous row and column permutations to obtain the general form given in (\ref{eq:subaction}); however, the general construction and the reason for its existence should be clear.

In group theoretical language, we now observe (using elementary matrix multiplication) that (\ref{eq:subaction}) provides a \emph{representation}\footnote{The meaning of this will be given in the proof of Theorem~\ref{thm:subrep}.} of the group $\times^m\text{Aff}(k\!-\!1)$, where $\times^m\text{Aff}(k\!-\!1)$ is realised by taking $m$ direct products of $\text{Aff}(k\!-\!1)$ via the Kronecker product $M'_A=M'_{s_1}\otimes M'_{s_2}\otimes \ldots \otimes M'_{s_m}$.
Ultimately, the dimensional reduction we are about to describe hinges on the existence of this representation and is very much in accord with our general philosophy that emphasizes the importance of the transformation rules (\ref{eq:leafactionP}) and (\ref{eq:leafactionG}), as well as our previous explorations of Markov invariants \cite{sumner2008markov,jarvis2014adventures}.

\begin{thm}
\label{thm:subrep}
The form (\ref{eq:subaction}) provides a faithful (``one to one'') representation of the direct product group $\times^m\emph{Aff}(k\!-\!1)$ with matrix entries that are linear in the matrix entries of the individual copies of $\emph{Aff}(k\!-\!1)$.  
\end{thm}
\begin{proof}
The composition rule for two matrices in the affine group $\text{Aff}(k\!-\!1)$ is revealed by the computation 
\[
\left(
\begin{matrix}
T & v \\
0 & 1
\end{matrix}
\right)
\left(
\begin{matrix}
U & w \\
0 & 1
\end{matrix}
\right)
=\left(
\begin{matrix}
TU & Tw+v \\
0 & 1
\end{matrix}
\right).
\]
In abstract terms, this can be understood by mapping each member $\left(\begin{smallmatrix}
T & v \\
0 & 1
\end{smallmatrix}
\right)
\in
\text{Aff}(k\!-\!1)$ to the pair $(T,v)$ and defining the group product ``$\ast$'' via the rule\footnote{In fact, this the natural way to \emph{define} $\text{Aff}(k\!-\!1)$ in the first place.}:
\beqn
\label{eq:affrule}
(T,v)\ast (U,w)=(TU,Tw+v).
\eqn
In the case of the group $\times^m \text{Aff}(k\!-\!1)$, each member may be represented as an $m$-tuple 
\[
\left((T_1,v_1),(T_2,v_2),\ldots,(T_m,v_m)\right),
\] 
and the product rule is provided by mimicking (\ref{eq:affrule}) in each entry of the tuple.

Correspondingly, the product of any two matrices of the form (\ref{eq:subaction}) is given by
\scriptsize
\begin{align*}
\left(
\begin{matrix}
T_{i_1} & 0  & 0 & \cdots & v_{i_1} \\
0 & T_{i_2}  & 0 & \cdots & v_{i_2}  \\
\vdots &  & \ddots  & & \vdots \\
0 & 0  & \cdots & T_{i_m} & v_{i_m} \\
0 & 0  & \cdots & 0 & 1\\
\end{matrix}
\right)
&
\left(
\begin{matrix}
U_{i_1} & 0  & 0 & \cdots & w_{i_1} \\
0 & U_{i_2}  & 0 & \cdots & w_{i_2}  \\
\vdots &  & \ddots  & & \vdots \\
0 & 0  & \cdots & U_{i_m} & w_{i_m} \\
0 & 0  & \cdots & 0 & 1\\
\end{matrix}
\right)
=
\left(
\begin{matrix}
T_{i_1}U_{i_1} & 0  & 0 & \cdots & T_{i_1}w_{i_1}+v_{i_1} \\
0 & T_{i_2}U_{i_2} & 0 & \cdots & T_{i_2}w_{i_2}+v_{i_2}  \\
\vdots &  & \ddots  & & \vdots \\
0 & 0  & \cdots & T_{i_m}U_{i_m} & T_{i_m}w_{i_m}+v_{i_m} \\
0 & 0  & \cdots & 0 & 1\\
\end{matrix}
\right).
\end{align*}
\normalsize
Inspection of this product then shows that the matrices (\ref{eq:subaction}) form what is known as a representation of $\times^m \text{Aff}(k\!-\!1)$, as required.

Faithfulness of the representation --- considered as an injective mapping from $\times^m \text{Aff}(k\!-\!1)$ into the set of matrices of the form (\ref{eq:subaction}), is obvious.

Linearity of the matrix entries in (\ref{eq:subaction})  is also obvious.

\end{proof}


Under the similarity transformation (\ref{eq:simtrans}), the distribution $P$ changes to (again expressed in tensor notation\footnote{Exactly how this affects the site pattern probabilities $p_{i_1i_2\ldots i_L}$ is given  in the  appendix (\ref{eq:Psim}).}):
\beqn
\label{eq:simP}
P\rightarrow \underbrace{S\otimes S\otimes \ldots \otimes S}_{L=m+n\text{ products}}\cdot P.
\eqn
Concurrently, the ``transformed flattening'' $\text{Flat}'_{A|B}(P)$, is defined as
\beqn
\text{Flat}'_{A|B}(P):&=\text{Flat}_{A|B}(\underbrace{S\otimes S\otimes \ldots \otimes S}_{L=m+n\text{ products}}\cdot P)
\\
&=\underbrace{\left(S\otimes S\otimes \ldots \otimes S\right)}_{m\text{ products}}\text{Flat}_{A|B}(P){\underbrace{\left(S\otimes S\otimes \ldots \otimes S\right)}_{n\text{ products}}}^T.
\nonumber
\eqn
Thus, the transformation rule (\ref{eq:leafactionG}) can then be expressed as
\beqn
\label{eq:leafactionS}
\text{Flat}'_{A|B}(P)\rightarrow M'_A\text{Flat}'_{A|B}(P){M'_B}^{\!\!T}.
\eqn

We now locate the subset of $(m(k\!-\!1)\!+\!1)$ rows and subset of $(n(k\!-\!1)\!+\!1)$ columns of $\text{Flat}'_{A|B}(P)$ corresponding to the rows of $\widehat{M}'_A$ and the columns of $\left.\widehat{M}'_B\right.^{\!\!T}$ respectively (details are given in the appendix (\ref{eq:locate})), and define the ``sub-flattening'' $\widehat{\text{Flat}}'_{A|B}(P)$ as a $(m(k\!-\!1)\!+\!1)\times (n(k\!-\!1)\!+\!1)$ sub-matrix of $\text{Flat}'_{A|B}(P)$, and observe:

\begin{cor}
\label{cor:trans}
The sub-flattening $\widehat{\emph{Flat}}'_{A|B}(P)$ satisfies the transformation rule
\beqn
\label{eq:subflattrans}
\widehat{\text{Flat}}'_{A|B}(P)\rightarrow \widehat{M}'_A\widehat{\text{Flat}}'_{A|B}(P)\left.\widehat{M}'_B\right.^{\!\!T}.
\eqn
\end{cor}
\begin{proof}
The result follows from (\ref{eq:leafactionS}) together with Theorem~\ref{thm:subrep} and the observation that the form $\widehat{M}'_A$ (\ref{eq:subaction}) is block upper-triangular, and similarly the transpose $\left.\widehat{M}'_B\right.^{\!\!T}$ is block lower-triangular.
\end{proof}

Strikingly, according to Theorem~\ref{thm:subrep}, the sub-flattening transforms \emph{bilinearly}  in the parameters of the Markov matrices acting at the leaves of the phylogenetic tree (as opposed to the multi-linear rule exhibited in  (\ref{eq:leafactionP}) and (\ref{eq:leafactionG})).
Additionally, we find that the contrasting rank conditions on $\text{Flat}_{A|B}(P)$ discussed in the previous section are retained in a modified form on the sub-flattening $\widehat{\text{Flat}}'_{A|B}(P)$.
To state this result, we first need:
\begin{mydef}
We say that a phylogenetic model on a tree $\mathcal{T}$ is \emph{generic} if the initial distribution $(\pi_i)_{i\in \kappa}$ at the root of the phylogenetic tree has no zero entries and the Markov matrix $M_{e}$ on each edge $e\in\mathcal{T}$ has full rank.
\end{mydef}

It follows that:
\begin{thm}
\label{thm1}
The sub-flattening $\widehat{\text{Flat}}'_{A|B}(P)$ has rank at most $r(k\!-\!1)\!+\!1$, where $r\geq 1$ is the parsimony score for the split $A|B$ considered as a binary character at the leaves of $\mathcal{T}$.
Up to the specified bound, maximal rank is attained in the generic case.
In particular, in the generic case, the sub-flattening $\widehat{\text{Flat}}'_{A|B}(P)$ has rank $k$ if the edge corresponding to the split $A|B$ occurs in $\mathcal{T}$ (since in this case $r\!=\!1$) and rank equal to or greater than $2(k\!-\!1)+1$ otherwise.
\end{thm}

\begin{proof}
The proof is given in Appendix~\ref{append}.
\end{proof}


These rank conditions are exactly analogous to the conditions for the flattenings given in \cite{eriksson2005tree} under the substitution $k^r\rightarrow r(\!k-\!1)\!+\!1$.
We hence propose to use the sub-flattenings as an alternative, practical test for the existence of specific edges in a phylogenetic tree.

\section{Discussion}
\label{sec:discuss}

Inspired by the previous approaches taken by \cite{allman2016split,casanellas2007performance,chifman2014quartet,eriksson2005tree,fernandez2015invariant} (as described at the end of the background section), we are currently exploring the computational means to exploit Theorem~\ref{thm1} in a practical setting.
A potential obstruction is the change of basis required to convert an observed site pattern array $F$ via the  similarity transformation $S$. 
However, we note that this may be efficiently achieved by only computing the entries required for the sub-flattening and also by summing over only the non-zero entries of $F$. 
For large number of taxa $L$, $F$ is an extremely large, and hence sparse, array.
Without observing that the required transformation can be achieved efficiently there simply would be no way of computing the sub-flattening efficiently and hence the approach described in this paper would be of no practical use.
From these observations, we claim this is not an insurmountable computational obstruction.

Additionally, in practical cases where the observed site pattern array $F$ is obtained from \emph{finite} sequence alignments, and hence is subject to standard statistical sampling error, it is quite possible that the specific choice of similarity transformation $S$ (beyond having a constant row) will have an effect upon the performance of the method.
Given this observation, and general best practice in numerical computation, it is likely that it is sensible to demand $S$ to be an orthogonal matrix.
At this stage we have not investigated this further (either theoretically or via simulation).

We suspect that the best way to test for low rank of the sub-flattenings will be via SVD, but at present it is unclear what the optimal numeric approach will be since the statistical properties under multinomial sampling of the algebraic methods described here are unclear.
Future work will explore these questions via simulations and testing on real data sets.
Nonetheless, the results presented in this paper are expected to lead to a novel, computationally efficient phylogenetic method consistent with the general Markov model of molecular state evolution.

\subsection*{Acknowledgments}
This work was inspired from a question Alexei Drummond put to Barbara Holland during her presentation at the New Zealand Phylogenetics Meeting, DOOM 2016.
I would also like to thank the anonymous reviewer for their careful and substantive comments that lead to a greatly improved manuscript. 

\subsection*{Funding}
This work was supported by the Australian Research Council Discovery Early Career Fellowship DE130100423.

\bibliographystyle{plain}
\bibliography{biblio}

\appendix
\numberwithin{equation}{section}
\section{Proof of Theorem~\ref{thm1}}
\label{append}

Our general approach to the proof will be to give conditions for when the rank of the sub-flattenings does or does not grow under phylogenetic divergence events.
In particular, we will show that the rank of the sub-flattenings is unchanged after a phylogenetic divergence event which is ``consistent'' with the split under consideration (the precise meaning of this will become evident below).
Although related, our proof method is different in conception from the approach taken in \cite{eriksson2005tree} for obtaining the analogous conditions for the ranks of the full flattenings.

\begin{mydef}
Given a rooted tree $\mathcal{T}$, consider the subtrees consisting of a vertex in $\mathcal{T}$ together with all of its descendants (including the case where the subtree consists of a leaf vertex only).   
Given a subset $A$ of leaves, we say such a subtree \emph{is $A$-consistent} if its leaves are a subset of $A$.
We say an $A$-consistent subtree is \emph{maximally $A$-consistent} if it is not itself a subtree of an $A$-consistent subtree.
Similarly, given a split $A|B$ we say that a subtree is  \emph{$A|B$-consistent} if its leaves are a subset of $A$ \emph{or} $B$; together with the corresponding definition of \emph{maximally $A|B$-consistent}.
\end{mydef}
An example is given in Figure~\ref{fig:subtrees}.

%
%
%
%
%
%
%

\begin{figure}[htb]
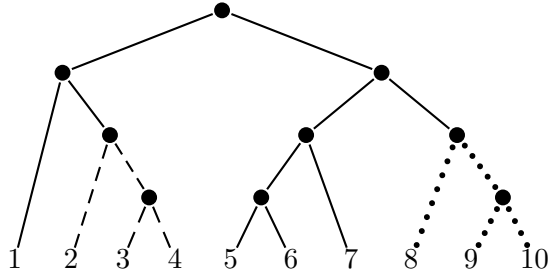

\caption{A rooted tree with two maximally $A|B=\{2,3,4,6\}|\{1,5,7,8,9,10\}$ consistent subtrees indicated .
The subtree with leaf set $\{3,4\}$ is $A$-consistent, but not maximally so.}
\centering
    \pstree[levelsep=5ex,treesep=.5cm,nodesep=.8pt]{\TC*[radius=3pt] }{%
    
     \pstree{\TC*[radius=3pt]}{
                \skiplevels{2}
                \TR{$1$}
                \endskiplevels
                \pstree{\TC*[radius=3pt]}{
                \psset{linestyle=dashed}
                \skiplevel{\TR{2}} 
                \pstree{\TC*[radius=3pt]}{
                \TR{3}
                \TR{4}
                } 
                }
            }

     \pstree{\TC*[radius=3pt] }{       
     	\pstree{\TC*[radius=3pt]}{
                \pstree{\TC*[radius=3pt]}{
                \TR{$5$}
                \TR{$6$} 
                } 
                \skiplevel{\TR{$7$}} 
                }
     	\pstree{\TC*[radius=3pt]}{
     		\psset{linestyle=dotted,linewidth=0.08}
        		\skiplevel{\TR{$8$}} 
                \pstree{\TC*[radius=3pt]}{
                \TR{$9$}
                \TR{10} 
                } 
                }
      }

       }
    \label{fig:subtrees}
\end{figure}

\begin{lem}
\label{lem:indep}
If $P$ is a pattern distribution arising from a tree $\mathcal{T}$ under the general Markov model, the rank of the sub-flattening $\widehat{\text{Flat}}'_{A|B}(P)$ is independent of the size and/or structure of any $A|B$-consistent subtrees of $\mathcal{T}$.
\end{lem}

\begin{proof}

Consider the molecular state space $\kappa\!=\!\{1,2,\ldots,k\}$ and a site pattern probability distribution $p_{i_1i_2\ldots i_L}$ on $L$ taxa.
Suppose this distribution arises under the general Markov model on the tree $\mathcal{T}$ and subsequently a time-instantaneous divergence event occurs causing, without loss of generality, a copy of the $L^\text{th}$ taxon to be created.
Under the usual assumptions of this model, this results in a new distribution $P^+=(p^+_{i_1i_2\ldots i_{L}i_{L+1}})_{i_j\in\kappa}$ on an $L\!+\!1$ taxon tree $\mathcal{T}^+$, with
\beqn
\label{eq:phylobranch}
p^+_{i_1i_2\ldots  i_{L}i_{L+1}}=
\left\{
\begin{array}{ll}
p_{i_1i_2\ldots  i_L}\text{ if }i_{L}\!=\!i_{L+1};\\
0,\text{ otherwise.}
\end{array}
\right.
\eqn

Consider a split $A|B$ and suppose taxon $L$ is contained in $B$.
Consider the new split $A|B'$ where  the new taxon $L\!+\!1$ has been adjoined to $B$ to produce $B'=B\cup \{L\!+\!1\}$. 
We will show that the sub-flattening $\widehat{\text{Flat}}'_{A|B'}(P^+)$ is obtained from $\widehat{\text{Flat}}'_{A|B}(P)$ by simply repeating $k\!-\!1$ columns.

Let $S$ be any $k\!\times\! k$ matrix consistent with the similarity transformation (\ref{eq:simtrans}).
In particular, this means that the $k^\text{th}$ row of $S$ is constant, and, without loss of generality, we will assume this is a row of $1s$, i.e. $S_{kj}\!=\!1$ for $j=1,2,\ldots, k$.
We denote the application of this similarity transformation to the site pattern distribution as
\beqn
\label{eq:Psim}
q_{i_1i_2\ldots i_L}:=\sum_{j_1,j_2\ldots,j_L\in \kappa}S_{i_1j_1}S_{i_2j_2}\ldots S_{i_Lj_L}p_{j_1j_2\ldots j_L}.
\eqn
We will refer to these quantities as the ``$q$-coordinates''.

Now suppose, without loss of generality, $A|B\!=\!\{1,2,\ldots,m\}|\{m+1,m+2,\ldots,L\}$, and write $q_{i_1i_2\ldots i_m,j_1j_2\ldots j_n}$ to emphasize the flattening corresponding to this split.
After locating the rows and columns which define the form (\ref{eq:subaction}),  the $(m(k\!-\!1)+1)\times (n(k\!-\!1)+1)$ entries of the sub-flattening are seen to be given by
\beqn
\label{eq:locate}
\widehat{\text{Flat}}'_{A|B}(P)
=
\left(
\begin{matrix}
\\
& q_{i_1i_2\ldots i_m,j_1j_2\ldots j_n} &
\\
\\
\end{matrix}
\right){\hspace{-.8em}\begin{smallmatrix}\vspace{3.2em}\\\hspace{.25em}i_a=k\text{ for all but at most one } a\in\{1,2,\ldots,m\};\\j_b=k\text{ for all but at most one } b\in\{1,2,\ldots,n\}.\end{smallmatrix}}
\eqn

We now consider the effect of the divergence rule (\ref{eq:phylobranch}) on the $q$-coordinates.
Again we suppose that the divergence event occurs on the $L^{\text{th}}$ taxon.
As a consequence of (\ref{eq:phylobranch}), a short computation shows that
\[
q^+_{i_1i_2\ldots i_m, j_{1}j_2\ldots j_{n-1}j_nj_{n+1} }
=\sum_{j,j'\in \kappa}S_{j_n j}S_{j_{n+1}j}S^{-1}_{jj'}q_{i_1i_2\ldots i_m,  j_1j_2\ldots j_{n-1}j'}.
\]

To construct $\widehat{\text{Flat}}'_{A|B}(P^+)$ we must consider three cases (recalling that we are assuming $S_{kj}\!=\!1$ for each $j\!=\!1,2,\ldots,k$):
\begin{enumerate}[(i.)]
\item Suppose $j_n\!=\!j_{n+1}\!=\! k$.
Then
\[
q^+_{i_1i_2\ldots  i_m,j_1j_2\ldots j_{n-1}kk }=q_{i_1i_2\ldots  i_m,j_1j_2\ldots j_{n-1}k}.
\]
\item Suppose $j_n\!=\!k$ and  $j_{n+1}\!\neq \! k$.
Then
\[
q^+_{i_1i_2\ldots  i_m,j_1j_2\ldots j_{n-1}kj_{n+1} }=q_{i_1i_2\ldots  i_m,j_1j_2\ldots j_{n-1}j_{n+1}}.
\]
\item Suppose $j_n\!\neq \!k$ and  $j_{n+1}\!=\! k$. Then
\[
q^+_{i_1i_2\ldots  i_m,j_1j_2\ldots j_{n-1}j_nk }=q_{i_1i_2\ldots  i_m,j_1j_2\ldots j_{n-1}j_{n}}.
\]
\end{enumerate}

In particular, for each choice $j=1,2,\ldots k\!-\!1$,
\[
q^+_{i_1i_2\ldots  i_m,\underbrace{\scriptstyle{kk\ldots kk}}_{n \text{ indices}}\hspace{-.17em}j }=q_{i_1i_2\ldots  i_m,\hspace{-.5em}\underbrace{\scriptstyle{kk\ldots k}}_{n\!-\!1 \text{ indices}}\hspace{-.7em}j}=q^+_{i_1i_2\ldots  i_m,\hspace{-.5em}\underbrace{\scriptstyle{kk\ldots k}}_{n\!-\!1 \text{ indices}}\hspace{-.75em}jk }.
\]
Comparing to the general form (\ref{eq:locate}), we see that the sub-flattening $\widehat{\text{Flat}}_{A|B'}'(P^+)$ is produced from the sub-flattening $\widehat{\text{Flat}}_{A|B}'(P)$ by simply repeating $k\!-\!1$ columns.
This observation holds more generally, independently of which taxon the divergence event occurs on.
The only modification needed is when the divergence happens on the left side of the split $A|B$, in which case the new sub-flattening is obtained from the old by a repetition of rows rather than columns. 
Thus, if we place a new taxon into the same side of the split as the taxon it diverged from, the rank of the sub-flattening is preserved.

We now apply Corollary~\ref{cor:trans} to conclude that, in the generic case, further application of (full rank) Markov matrices at the leaves of the phylogenetic tree $\mathcal{T}^+$ also does not affect the rank of the sub-flattening.

These observations establish the lemma.

\end{proof}


%
%


Now in order to determine the rank of an arbitrary sub-flattening, we may repeatedly apply Lemma~\ref{lem:indep} to reduce to the case where each $A|B$-consistent subtree is a single leaf.
Assuming this situation, each leaf is then either (i) not part of a cherry, or (ii) part of a cherry where the two leaves in the cherry lie on complementary sides of the split $A|B$.
A key feature of this situation is that we can label the descendants of every vertex (excluding the root) with complementary binary labels such that the leaf labels are consistent with the split $A|B$.
For our purposes, we then consider this reduced case as arising from a sequence of divergence events from the base two-taxa case where, after each divergence event at a leaf, the two descendants are placed into complementary sides of the target split $A|B$.
An example illustrating that this process is always possible given in Figure~\ref{fig:subclippings}.

\begin{figure}[htb]
\captionsetup{singlelinecheck=off}
\caption[]{
Given a tree $\mathcal{T}$ and a split $A|B$ on its leaf set, leaves belonging to $A$ are labelled by ``$+$'' and leaves in $B$ are labeled by ``$-$''.
The tree is reduced by removing any $A|B$-consistent subtrees, and  binary labels are attached to the vertices (excluding the root) such that the descendants of each vertex obtain complementary labels and the leaf labels are consistent with the split $A|B$.
In the case illustrated, the second step follows as a consequence of the two leaves that are not part of cherry.
  \begin{itemize}
    \item[\textbf{Step 1.}] Reduce each maximally $A|B$ consistent subtree to a leaf. 
    \item[\textbf{Step 2.}] Label each internal vertex (excluding the root) consistently so descendants of internal vertices are distinctly labelled.
    \item[\textbf{Step 3.}] Arbitrarily resolve any remaining ambiguities.
  \end{itemize}}
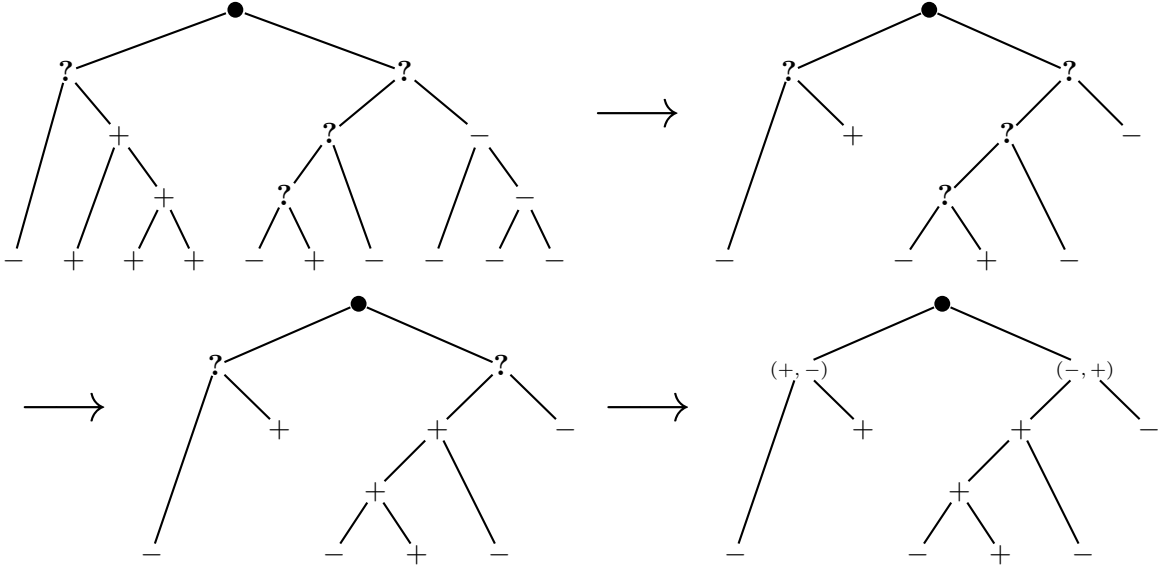
\begin{tabular}{lll}
    \pstree[levelsep=5ex,treesep=.5cm,nodesep=.8pt]{\TC*[radius=3pt] }{%
    
     \pstree{\TR{\textbf{?}}}{
                \skiplevels{2}
                \TR{$-$}
                \endskiplevels
                \pstree{\TR{$+$}}{
                \skiplevel{\TR{$+$}} 
                \pstree{\TR{$+$}}{
                \TR{$+$}
                \TR{$+$} 
                } 
                }
            }

     \pstree{\TR{\textbf{?}} }{       
     	\pstree{\TR{\textbf{?}}}{
                \pstree{\TR{\textbf{?}}}{
                \TR{$-$}
                \TR{$+$} 
                } 
                \skiplevel{\TR{$-$}} 
                }
     	\pstree{\TR{$-$}}{
        		\skiplevel{\TR{$-$}} 
                \pstree{\TR{$-$}}{
                \TR{$-$}
                \TR{$-$} 
                } 
                }
      }

       }

	\huge{       
       $
       \begin{matrix}
       \vspace{-4.5em}
       \longrightarrow
       \end{matrix}
       $
       }
       
       \normalsize
    
 \pstree[levelsep=5ex,treesep=1.8cm,nodesep=.4pt]{\TC*[radius=3pt] }{
    
    	\skiplevels{0}
    	\pstree[levelsep=5ex,treesep=1.4cm,nodesep=.4pt]{\TR{\textbf{?}}}{
        	\skiplevels{2}
            \TR{$-$}
            \endskiplevels
            \TR{$+$}
        }
    	\endskiplevels
    
    	\pstree[levelsep=5ex,treesep=1.4cm,nodesep=.4pt]{\TR{\textbf{?}}}{
        	\pstree[levelsep=5ex,treesep=.8cm,nodesep=.4pt]{\TR{\textbf{?}}}{
            	\pstree[levelsep=5ex,treesep=.8cm,nodesep=.4pt]{\TR{\textbf{?}}}{
            		\TR{$-$}
                	\TR{$+$}
            	}
                \skiplevel{\TR{$-$}}
             }
             
             \skiplevels{0}
            \TR{$-$}
            \endskiplevels
         }

    } 
    
\\
    
\huge{       
       $
       \begin{matrix}
       \vspace{-4.5em}
       \longrightarrow
       \end{matrix}
       $
       }
       
       \normalsize
    
    \pstree[levelsep=5ex,treesep=1.8cm,nodesep=.4pt]{\TC*[radius=3pt] }{
    
    	\skiplevels{0}
    	\pstree[levelsep=5ex,treesep=1.4cm,nodesep=.4pt]{\TR{\textbf{?}}}{
        	\skiplevels{2}
            \TR{$-$}
            \endskiplevels
            \TR{$+$}
        }
    	\endskiplevels
    
    	\pstree[levelsep=5ex,treesep=1.4cm,nodesep=.4pt]{\TR{\textbf{?}}}{
        	\pstree[levelsep=5ex,treesep=.8cm,nodesep=.4pt]{\TR{+}}{
            	\pstree[levelsep=5ex,treesep=.8cm,nodesep=.4pt]{\TR{$+$}}{
            		\TR{$-$}
                	\TR{$+$}
            	}
                \skiplevel{\TR{$-$}}
             }
             
             \skiplevels{0}
            \TR{$-$}
            \endskiplevels
         }

    } 
 
\huge{       
       $
       \begin{matrix}
       \vspace{-4.5em}
       \longrightarrow
       \end{matrix}
       $
       }
       
       \normalsize

\pstree[levelsep=5ex,treesep=1.8cm,nodesep=.4pt]{\TC*[radius=3pt] }{
    
    	\skiplevels{0}
    	\pstree[levelsep=5ex,treesep=1.4cm,nodesep=.4pt]{\TR{\tiny{$(+,-)$}}}{
        	\skiplevels{2}
            \TR{$-$}
            \endskiplevels
            \TR{$+$}
        }
    	\endskiplevels
    
    	\pstree[levelsep=5ex,treesep=1.4cm,nodesep=.4pt]{\TR{\tiny{$(-,+)$}}}{
        	\pstree[levelsep=5ex,treesep=.8cm,nodesep=.4pt]{\TR{$+$}}{
            	\pstree[levelsep=5ex,treesep=.8cm,nodesep=.4pt]{\TR{$+$}}{
            		\TR{$-$}
                	\TR{$+$}
            	}
                \skiplevel{\TR{$-$}}
             }
             
             \skiplevels{0}
            \TR{$-$}
            \endskiplevels
         }

    }

   \end{tabular}
    \label{fig:subclippings}
\end{figure}

We use this process to establish:
\begin{lem}
\label{lem:fullrank}
Suppose $\mathcal{T}$ is a tree, suppose $P$ is a distribution arising on $\mathcal{T}$ under the general Markov model, and suppose $A|B$ is a split such that the maximally $A|B$-consistent subtrees are all leaves. 
Then, in the generic case, the sub-flattening $\widehat{\text{Flat}}'_{A|B}(P)$ has maximal rank.
\end{lem}
\begin{proof}
Suppose such a reduced tree has $q$-coordinates $q_{i_1i_2\ldots i_m,j_1j_2\ldots j_n}$, and the $n^{\text{th}}$ taxon in $B$ diverges creating a new taxon which is adjoined to $A$ to form the new split $A'|B$.
Analogous to the previous situation, we have the new $q$-coordinates
\[
q^+_{i_1i_2\ldots i_mi_{m+1},j_1j_2\ldots j_n}=\sum_{j,j'=1}^kS_{i_{m+1}j}S_{j_{n}j}S_{jj'}^{-1}q_{i_1i_2\ldots i_{m},j_1j_2\ldots j_{n-1}j'}.
\]

From this we see that the additional $k\!-\!1$ rows in the sub-flattening $\widehat{\text{Flat}}'_{A'|B}(P^+)$ are obtained by setting $i_1\!=\!i_2\!=\ldots =\!i_m\!=k$, and taking $i_{m+1}\!=\!1,2,\ldots,k\!-\!1$ in 
\[
q^+_{kk\ldots ki_{m+1},j_1j_2\ldots j_n}=\sum_{j,j'=1}^kS_{i_{m+1}j}S_{j_{n}j}S_{jj'}^{-1}q_{kk\ldots k,j_1j_2\ldots j_{n-1}j'},
\]
where the columns are indexed by choosing $b\in \{1,2,\ldots,n\}$ so that at most a single $j_b\neq k$ at a time.
In particular, if we choose $j_1\neq k$ and $j_2\!=\!j_3\!=\!\ldots \!=\!j_n=k$ we have
\[
q^+_{kk\ldots ki_{m+1},j_1kk\ldots k}=q_{kk\ldots k,j_1kk\ldots ki_{m+1}}.
\]
Now for each choice $i_{m+1}\!=\!1,2,\ldots, k\!-\!1$ this expression gives $q$-coordinates which \emph{do not} appear in the sub-flattening $\widehat{\text{Flat}}_{A|B}'(P)$ or any of the other rows of $\widehat{\text{Flat}}'_{A'|B}(P^+)$.
It follows that any linear dependencies between the new and remaining rows in $\widehat{\text{Flat}}'_{A'|B}(P^+)$ would imply linear constraints on the $q$-coordinates on the original $m+n$ taxon tree.
In turn, this would imply the existence of linear phylogenetic invariants for the general Markov model, which are known not to exist \cite{Hagedorn00acombinatorial}.
Therefore, the new rows appearing in $\widehat{\text{Flat}}'_{A'|B}(P^+)$ are linearly independent from the rest.

To complete the proof, we use induction on the base case of a two-taxon tree.
To establish this base case, we show that, in the generic case, the two-taxon sub-flattening on the split $A|B=\{1\}|\{2\}$ has full rank $k\!=\!(k\!-\!1)\!+\!1$.
This follows easily since, in the two-taxon case, the sub-flattening is equal to the transformed flattening, that is
\[
\widehat{\text{Flat}}'_{\{1\}|\{2\}}(P)=\text{Flat}'_{\{1\}|\{2\}}(P)=S\text{Flat}(P)_{\{1\}|\{2\}}S^{-1}.
\]
Thus the sub-flattening is related by the similarity transformation $S$ to the flattening $\text{Flat}(P)_{\{1\}|\{2\}}$, which a standard argument shows can be expressed as
\[
\text{Flat}_{\{1\}|\{2\}}(P)=M_1D(\pi)M_2^T,
\]
where $D(\pi)$ is the diagonal matrix formed from the root distribution $\pi=(\pi_i)_{i\in \kappa}$.
Clearly this matrix is full rank if $M_1$ and $M_2$ are full rank and $\pi$ has no zero entries. 
Thus, in the generic case, the two-taxon sub-flattening $\widehat{\text{Flat}}'_{1|2}(P)$ has full rank.

Induction on this base case establishes the lemma.

\end{proof}

With these results in hand, Theorem~\ref{thm1} follows for arbitrary trees and splits by the following three steps:
\begin{itemize}
\item[(i)] Apply Lemma~\ref{lem:indep} and clip off any $A|B$-consistent subtrees;
\item[(ii)] Apply Lemma~\ref{lem:fullrank}; and
\item[(iii)] Use Fitch's algorithm \cite{felsenstein2004inferring,fitch1971toward} to recognise that the minimum of the number of maximally $A$- and $B$-consistent subtrees is none other than the parsimony score for the split $A|B$ considered as a binary character at the leaves of $\mathcal{T}$.
\end{itemize}

\end{document}